\documentclass[12pt,draftclsnofoot, onecolumn]{IEEEtran}
\usepackage{graphicx}
\DeclareGraphicsExtensions{.eps,.pdf,.jpg}
\usepackage{tikz}
\usetikzlibrary{arrows}
\usepackage{fancyhdr}
\usepackage{extramarks}
\usepackage{amsmath}
\usepackage{amsthm}
\usepackage{amsfonts}
\usepackage{amssymb}
\usepackage{algorithm}
\usepackage{algpseudocode}
\usepackage{bm}
\usepackage{soul}
\usepackage{color}
\usepackage{optidef}
\usepackage{enumitem}

\usepackage[justification=centering]{caption}
\newtheorem{theorem}{Theorem}
\newtheorem{lemma}{Lemma}
\newcommand{\RNum}[1]{\uppercase\expandafter{\romannumeral #1\relax}}

\renewcommand{\qedsymbol}{$\blacksquare$}

\title{Efficient Wireless Power Transmission}
\begin{document}
\title{Optimal and Near-Optimal Policies for Wireless Power Transfer in Energy-Limited and Power-Limited Scenarios}
%
%
%

\author{\IEEEauthorblockN{Roohollah~Rezaei\IEEEauthorrefmark{1},  Mohammad~Movahednasab\IEEEauthorrefmark{1}, Naeimeh Omidvar\IEEEauthorrefmark{1} and Mohammad Reza Pakravan\IEEEauthorrefmark{1}} \\
\IEEEauthorblockA{\IEEEauthorrefmark{1}Department of Electrical Engineering, Sharif University of Technology, Tehran, Iran.}
}
\maketitle

\begin{abstract}
Radio frequency wireless power transfer (RF-WPT) is an emerging technology that enables transferring energy from an energy access point (E-AP) to multiple energy receivers (E-Rs), in a wireless manner. In practice, there are some restrictions on the power level or the amount of energy that the E-AP can transfer, which need to be considered in order to determine a proper power transfer policy for the E-AP. In this paper, we formulate the problem of finding the optimal policy for two practical scenarios of power-limited and energy-limited E-APs. The formulated problems are non-convex stochastic optimization problems that are very challenging to solve. We propose optimal and near-optimal policies for the power transfer of the E-AP to the E-Rs, where the optimal solutions require statistical information of the channel states, while the near-optimal solutions do not require such information and perform well in practice. Furthermore, to ensure fairness among E-Rs, we propose two fair policies, namely Max-Min Fair policy and quality-of-service-aware Proportional Fair policy. MMF policy targets maximizing the minimum received power among the E-Rs, and QPF policy maximizes the total received power of the E-Rs, while guaranteeing the required minimum QoS for each E-R. Various numerical results demonstrate the significant performance of the proposed policies.
\end{abstract}


\section{Introduction}
Wireless power transfer (WPT) is considered as a key enabling technology for prolonging the lifetime of wireless networks. In many applications such as sensor networks, recharging batteries of  wireless nodes is a costly and time-consuming process. Moreover in some applications, such as medical implants inside human bodies, replacement of batteries is highly difficult and almost impractical.  To overcome these difficulties, WPT is proposed as a promising approach that provides continuous, stable and controllable energy resource  to wireless devices over the air \cite{Kang2015}. A similar concept to WPT is energy harvesting, which enables scavenging energy from ambient resources such as solar, wind and radio signals. However, the key advantage of WPT over energy harvesting is stability and controllability of the energy source.

In general, there are two types of WPT, regarding the technology behind. The first type is based on magnetic induction, in which energy is transferred from an E-AP to E-Rs by inductive coupling or magnetic resonance. The second type uses the radio frequency (RF) to transfer energy from one place to another. The latter approach has several advantages over the former approach, as follows. First of all, using radio frequency waves, WPT can be combined with wireless information transfer (WIT). Moreover, it covers longer transmission ranges, requires simpler receiver structure and better supports multiple receivers than the magnetic approach   \cite{Mou2015}. Due to the aforementioned benefits, throughout this paper, we mainly focus on the RF type of WPT systems.  Three different scenarios can be considered for information and energy transfer using RF transmission. The first scenario, known as radio frequency wireless power transfer (RF-WPT), considers transfer of power from an access point (AP) to one or more receivers \cite{Yang2014}-\cite{Zeng2015}. 
 The second scenario, known as simultaneous wireless information and power transfer (SWIPT),  considers transferring both power and information to the receivers \cite{Varshney2008}-\cite{Nasir2013}. In the third scenario, known as wireless powered communication network (WPCN),  the AP transmits both  power and information to the receivers who then use the received power for their uplink information transmission toward the AP \cite{Ju2014a}-\cite{Ju2014c}.  Similar to the previous works in \cite{Yang2014}-\cite{Zeng2015}, in this paper, we focus on the first scenario, i.e., the wireless power transfer. It should be noted that our proposed approaches for the power transfer can be easily extended to the case where both power and information is considered. 

In most RF-WPT systems, the energy access point (E-AP) uses beamforming to concentrate its transferred energy toward the energy receivers (E-Rs). For this purpose, the E-AP needs to know the channel state information (CSI) of the link toward each receiver to do beamforming. Several efforts have been done in the literature for estimating the channel coefficients \cite{Yang2014}-\cite{Zeng2015}. For example, Yang et al. in \cite{Yang2014} proposed a CSI estimation method in which  E-AP sends a known training sequence to the E-R, and then,  the E-R feeds back the observed state of the channel in the uplink to the E-AP. 
However, this procedure requires complex computations at the receivers side,  which may not be feasible in some applications where the computational resources of the receivers are very limited, e.g., wireless sensor networks. Lee, et al. \cite{Lee2016} and Xu, et al. \cite{Xu2014} proposed some low complexity methods based on a one-bit feedback scheme for estimating the channel via the receivers. Zeng, et al. \cite{Zeng2015} assumed that the channels between the E-AP and E-Rs are reciprocal and proposed that E-Rs send training symbols for the E-AP to estimate the channel.

It should be noted that most of the existing works in the literature focus on short-term optimization of the policies. In such scenarios, the time horizon is divided into small timeslots with fixed lengths equal to the channel coherence time, and the CSI of the channels at each timeslot is regarded as a fixed deterministic parameter in the formulation of the optimization problem, which is estimated at the beginning of each timeslot. The problem of finding  the optimal policy is then formulated as a deterministic optimization problem which is solved at each timeslot, independently.
As a consequent of such short-term solutions, the resources of each timeslot are not preserved 

Such short-term solutions lack a global view of the long-term CSI, and do not incorporate the long-term channel fluctuations in the transmission policy. For example, if the situation of a channel is poor in a timeslot, it cannot preserve the resources of this timeslot for better and more effective use in the upcoming timeslots which may have better CSI. Consequently, such short-term solutions cannot provide an efficient transmission policy for the network in a long-term average sense.
In contrast, by considering the long-term optimization of the policy, the energy resources of the E-APs can be utilized efficiently by avoiding transmission in the case of poor channel conditions and saving the energy to be used for transmission in the later timeslots when the channel condition is better. 


There are  few works on long-term optimizations  in the related literature. Dong, et al. \cite{Dong2016} considered a scenario in which one single-antenna AP transfers energy and information to multiple single-antenna receivers. They aimed at minimizing the average transmitted power subject to energy constraints and stability constraints  of the information queues in the AP. They then used a Lyapunov approach and proposed  a near-optimal solution to the problem. However, they only considered a single-antenna scenario for the AP and assumed that the channels between the AP and each of the receivers are totally separate (i.e., no interference is considered between them).

 Choi, et al. \cite{Choi2015} considered a scenario in which  the AP transfers energy to multiple receivers in the downlink, where each receiver has an uplink queue of the information waiting for receiving enough energy so as to transmit in the uplink when their channel condition is good enough.   The authors used a Lyapunov approach to minimize the transmitted energy subject to the stability of the information queues, and proposed a near-optimal solution.  Furthermore, Biason, et al. \cite{Biason2016} considered an AP that transfers energy to two receivers and receives their uplink information as well. The authors used Markov decision theory to maximize the minimum received information rate of the receivers. However, their proposed method requires  knowing the explicit  of the channel state distribution to obtain the optimal solution.

It should be noted that RF-WPT systems are considered as an important main component in SWIPT and WPCN systems as well. Consequently, finding the optimal energy transfer policy in RF-WPT scenario is an important optimization problem, which not only helps to efficiently design RF-WPT systems, but also is an important primary milestone for solving the corresponding problems in SWIPT and WPCN systems. Therefore, in this paper, we focus on RF-WPT systems and address optimal fair policies for RF-WPT. We consider an RF-WPT scenario in which an E-AP transfers energy to multiple E-Rs. The E-AP is equipped with multiple antennas and performs beam-forming to concentrate energy toward each E-R.  

We investigate both power-limited and energy-limited cases for the energy source of the E-AP. In the power-limited case where E-AP is connected to an electrical grid we aim to maximize the  average total received energy of the E-Rs subject to a maximum power budget at the E-AP. Moreover, in  the energy-limited case, the E-AP is connected to a battery with limited energy. Solar panels in rural areas, which
harvest the sun's energy during day time to provide electricity during night time  are good examples for energy-limited E-APs \cite{Akinyele2015}.
In this case, we aim to minimize the average  transmitted energy of the E-AP while providing the required received energy of each E-R.

 In the power-limited case, as the power budget is limited,  maximizing the total received energy of the E-Rs may lead to severe unfairness among them. This is because the E-AP needs to transmit less energy toward the nearer E-Rs than the farther E-Rs  to deliver the same amount of required energy to them, and hence, the E-AP tends to serve the nearer E-Rs only. This phenomena  is known as the  near-far problem \cite{Bi2015}. In order to maintain fairness among E-Rs,  previous works such as \cite{Biason2016} have focused on transferring equal amounts of energy to the E-Rs.


However, this results in  severe degradation in the performance of the whole system (in terms of the total transmit power) when there exists an E-R that is too far from the E-AP comparing to the other E-APs.  To alleviate the aforementioned fairness issue we propose two fairness models and our second model   obtains a reasonable performance, in addition to providing a fair distribution of power among E-Rs.

In the first model, which is called \textit{Max-Min Fairness} (MMF), we maximize the minimum received power among E-Rs. It is a typical fairness model which tries to transfer equal power to receivers irrespective of the distance of the receivers from the E-AP. Our algorithm does not need the distribution of CSI and its solution is applicable to several E-Rs compared to  \cite{Biason2016}.
In the second model we consider the sum of the logarithm function of the received energy for each E-R, known as \textit{ proportional fairness} \cite{Kelly1998}. In this model, the total utility increases if we decrease an amount of energy from a near E-R (which receives more power) and add the same amount to a farther E-R. The amount of increase in utility is proportional to the unfairness among E-Rs. Since the transmitter has to consume more power to transfer energy to farther receivers,  proportional fairness attains a trade-off between fairness and performance. In addition, in this model, we guarantee a minimum power for each receiver to provide the minimum required power of it. These two fairness models alleviate the severe unfairness among E-Rs and the second model tries to provide reasonable performance in scenarios where some E-Rs are far away from the E-AP, compared to others.          

In this paper, we formulate the aforementioned scenarios and propose novel stochastic optimization formulations for each scenario. Then, using some stochastic optimization techniques, we propose optimal and near-optimal solutions for the formulated problems. In more details, we first focus on the energy-limited case and derive an optimal policy for energy transfer from the E-AP to a single E-R (which both are equipped with multiple antennas to transmit and receive, respectively). The optimal solution requires some information (that will be discussed later) on the distribution of the CSI. Such information may not be available in practice. Therefore, we then propose a Min Drift Plus Penalty (MDPP) algorithm based on Lyapunov Optimization theory \cite{Neely2010}, which does not require to know the CSI distribution and is shown that attains a near-optimal solution. Next, we focus on the power-limited case, and derive optimal and near-optimal energy transfer policies.
 Finally, we focus on the near-far problem and use the MDPP algorithm to obtain a near-optimal solution of the two fairness models mentioned before. The main contributions of this paper can be summarized as follows:
\begin{itemize}
\item
The power-limited and energy-limited WPT problems are formulated with novel stochastic optimization problems.
\item
For each of the power-limited and energy-limited WPT scenarios, a closed-form expression for the optimal solution is derived. Moreover, near-optimal policies, which do not require any information on the CSI distribution, are also proposed.
\item
Furthermore, to ensure fairness among E-Rs, various fairness models are considered and near-optimal energy transfer policies are proposed for the formulated stochastic optimization problems.
\end{itemize}

The rest of the paper is organized as follows.  Section \RNum{2} introduces the system model and problem formulations. The proposed solutions for the energy-limited and power-limited WPT cases are described in Section \RNum{3} and \RNum{4}, respectively. The considered fairness models and their associated proposed solutions are presented in Section  \RNum{5}. Numerical results are presented in Section \RNum{6}, and finally, Section \RNum{7} concludes the paper.
%

\section{System Model and Problem Formulation} \label{sec:SysMod}
Consider a network with one E-AP and $K$ energy receivers, as shown in Fig. \ref{fig:SysMod} . The E-AP and E-Rs  are  equipped with $N$  and $M$ antennas, respectively, where $N >M$.
 The E-AP transfers energy to the E-Rs by transmitting a tone signal (for the sake of saving bandwidth), and employs  beamforming in order to  focus the transmit energy  toward  each E-R.
\begin{figure}
\centering
\includegraphics[height = 0.3\linewidth]{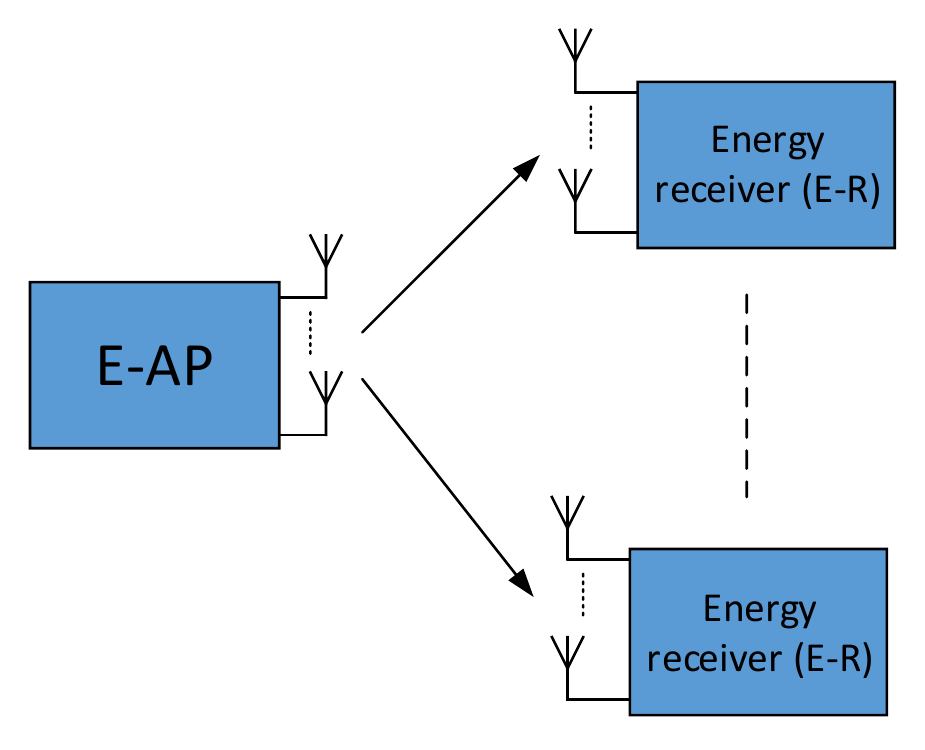}
\caption{System model}
\label{fig:SysMod}
\end{figure}
\begin{figure}
\centering
\includegraphics[height = 0.25\linewidth]{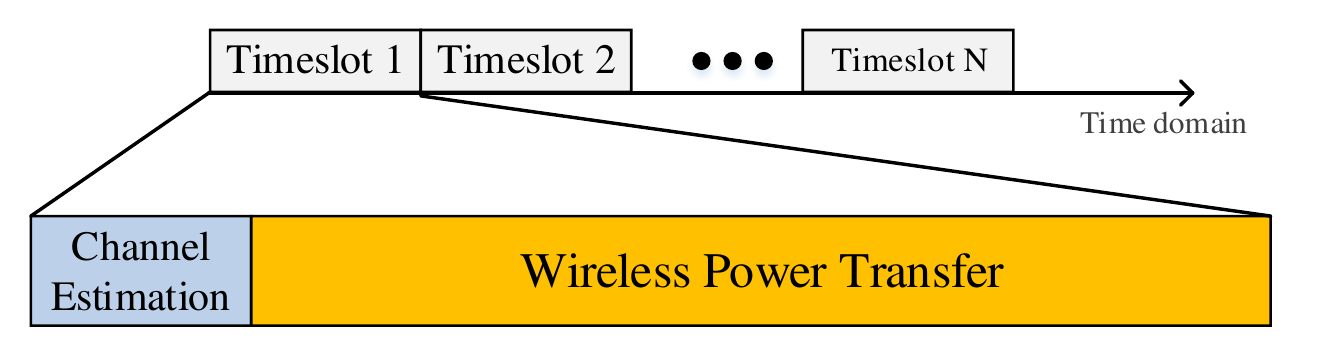}
\caption{Time-slotted system}
\label{fig:timeslotView}
\end{figure}

We consider a time-slotted system in which the  time domain is divided into timeslots of fixed length equal to the coherence time of the channels, as illustrated in Fig. \ref{fig:timeslotView}. As shown in Fig. \ref{fig:timeslotView}, at the beginning of each timeslot, a small portion of the timeslot is reserved for estimating the CSI of the outgoing channels by the E-AP.  The E-AP then uses  the rest of the time-slot for wireless power transfer.  Moreover, same as in many standard channel models \cite{Yang2014}-\cite{Ju2014c}, we consider a quasi-static flat-fading channel model for the channels between the E-AP and the E-Rs.  Note that the assumption of being quasi-static implies that the CSI of the channels remain constant during each time-slot and vary from one time-slot to the next one. 

%

We consider the MIMO Rician fading channel model \cite{Zeng2015}, in which the equivalent baseband channel between the E-AP and each E-R in the $l$-th timeslot is modeled as a complex matrix $\mathbf{H}_i[l]$ for each E-R. The $(m,n)$ entry of the channel matrix represents attenuation and delay for the link between   $m$-th antenna of  receiver and  $n$-th antenna of E-AP. This channel matrix remains constant during a timeslot and follows an independent identical distribution in successive timeslots.

%
%

\subsection{E-AP Transmission}   
{\color{black}
In each timeslot, the E-AP transmits a constant amplitude tone signal. The amplitude and phase of transmitted tone from each antenna is determined by entries of the  beamforming vector $\mathbf{x}[l] \in \mathbb{C}^{N\times1}$. Based on the described channel model, the received signal in the receiver $i$ is given by
\begin{equation}
\mathbf{y}_i(t) = \mathbf{H}_i[l]\mathbf{x}[l]+\mathbf{z}_i(t),\ \ lT \le t < (l+1)T,
\end{equation}
where $\mathbf{y}_i\in \mathbb{C}^{M\times1}$ denotes the baseband signal of receiver $i$,  and $\mathbf{z}_i \in \mathbb{C}^{M\times 1}$ represents the noise at receiver $i$. 

In each timeslot, the  E-AP chooses beamforming vector following a transmission policy. A transmission policy may be a function of current and/or previous channel states and  transmission history, also it may be a deterministic or probabilistic function.   In this paper we deal with policies which are deterministic functions of channel state and transmission history, but we show that they are optimal or near optimal among all possible policies. }

%

\subsection{E-R Reception}
The receivers use a rectifier to convert the received RF signal to a DC current, as shown in Fig. \ref{fig:RecvStr}. This current charges the batteries of the receivers.  The amount of harvested energy in receiver $i$, in a single timeslot, is denoted by $Q_i(t) = \zeta T\lVert\mathbf{y}_i(t)\rVert^2$, where $\zeta \in [0,1]$ models the efficiency in energy conversion. Without loss of generality, in this paper, we assume that $\zeta$ and $T$ equal one. We neglect the energy contribution of noise, as in \cite{Bi2014}, then we have

\begin{equation}
Q_i[l] \approx\left\lVert \mathbf{H}_i[l]\mathbf{x}[l]  \right \rVert^2 = Tr(\mathbf{W}_i[l]\mathbf{x}[l]\mathbf{x^\ast}[l]),
\end{equation}
where $\mathbf{W}_i[l]=\mathbf{H}_i^\ast[l] \mathbf{H}_i[l]$ and $Tr(\mathbf{A})$ is the trace of square matrix $\mathbf{A}$.

\begin{figure}
\centering
\includegraphics[height = 0.4\linewidth]{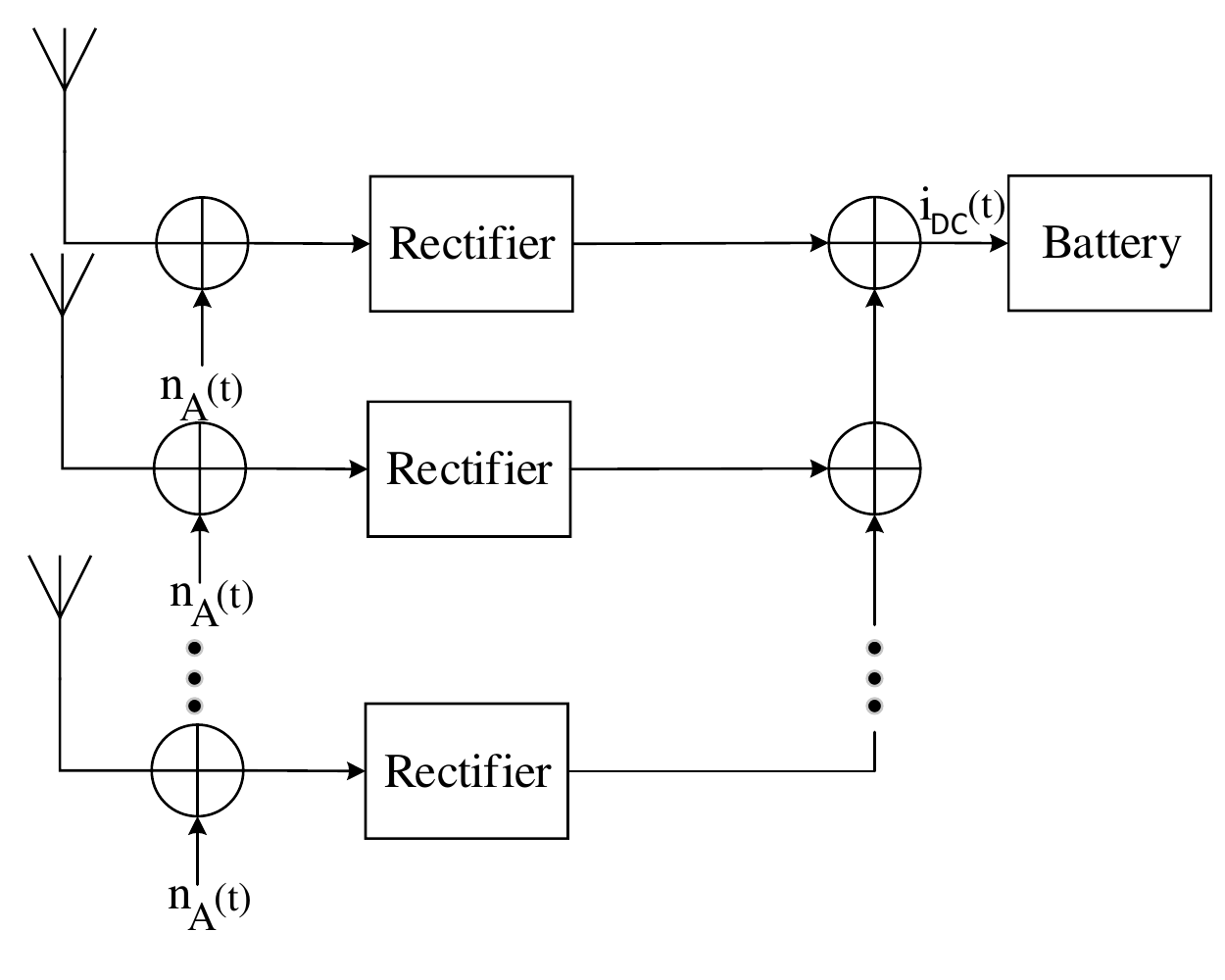}
\caption{Receiver structure}
\label{fig:RecvStr}
\end{figure}


\subsection{Long-Term Parameters}
We focus on long-term energy transfer optimization problems. The transmitted energy from the E-AP in timeslot $l$ equals $Tr(\mathbf{x}[l]\mathbf{x}^\ast[l])$,  hence the  expected value of the time averaged transmitted power in long term is as follows:
\begin{equation}
\bar{Q}_{AP} = \lim_{L \rightarrow \infty}\frac{1}{L}\sum_{l=0}^{L-1}\mathbb{E}[Tr(\mathbf{x}[l]\mathbf{x}^\ast[l])].
\end{equation}
It should be noted that $\mathbb{E}[.]$ denotes the average operator, which in the above equation is with respect to the randomness of the channel and policy (for policies with randomness) .
%
%
Similarly the expected value of the time averaged received power in receiver $i$ equals,
\begin{equation}
\bar{Q}_i = \lim_{L \rightarrow \infty}\frac{1}{L}\sum_{l=0}^{L-1}\mathbb{E}[Tr(\mathbf{W_i}[l]\mathbf{x}[l]\mathbf{x}^\ast[l])],\quad  \forall i = 1,\ldots,K.
\end{equation}

\section{Minimizing Average Transmitted Energy in Energy-limited Case}
\label{sec:MinTran}
In this section we consider a WPT system consisted of one E-AP that has limited energy budget (i.e., battery-operated E-AP)  and multiple E-Rs that require a minimum level of received power for their normal operations. A popular application example of  this scenario is an E-AP equipped with solar panels in a rural area which charges its batteries during the day and  transmits the stored energy toward the E-Rs during the night. 
An optimal transmitting policy  aims to minimize the average transmitted power of the E-AP  so as to maximize its lifetime, while at the same time, satisfying  the minimum power requirement of the each E-R. Accordingly, finding the optimal transmitting policy can be formulated by the following stochastic optimization problem:


\begin{mini!}|l|
 {\{\mathbf{x}(\mathbf{H})\} }{ \bar{Q}_{AP} = \lim_{L \rightarrow \infty}\frac{1}{L}\sum_{l=0}^{L-1}\mathbb{E}\left[Tr(\mathbf{x}[l]\mathbf{x}^\ast[l])\right]}{\label{P1}}{}
  \addConstraint{   \bar{Q}_i = \lim_{L \rightarrow \infty}\frac{1}{L}\sum_{l=0}^{L-1}\mathbb{E}[Tr(\mathbf{W}_i[l]\mathbf{x}[l]\mathbf{x}^\ast[l])]\ge P_i,\quad \forall  i = 1,\ldots,K  }\label{equ:P1ExpConst} 
  \addConstraint{{Tr(\mathbf{x}[l]\mathbf{x}^\ast[l])}\le P_{peak},\quad \forall  l \ge 0, }\label{equ:P1peak} 
 \end{mini!}\color{black}
where constraint \eqref{equ:P1ExpConst} guaranties the minimum power requirement of the E-Rs ($ P_i,\ \forall  i = 1,\ldots,K$), and constraint \eqref{equ:P1peak} is due to the limitation on the peak transmission power, $P_{peak}$, of the E-AP.
%

The optimization problem \eqref{P1} is highly non-trivial and includes some challenges that need to be addressed properly. First, due to constraint \eqref{equ:P1ExpConst}, the problem is non-convex. Moreover, it is a stochastic optimization problem that do not have closed-form expression for the objective function as well as constraint \eqref{equ:P1ExpConst}. In the rest of this section,  we first propose an optimal solution for the special case of single E-R. This solution provides a useful insight for finding  the optimal policy for the general case of multiple E-Rs. Then, we propose a transmission policy that does not require the knowledge of channel statistics, and determines the beamforming vector in each timeslot  based on  the instantaneous CSI of that timeslot  and the past transmission history. We analyze the performance of this algorithm and show that it satisfies the constraints of problem  \eqref{P1} and its performance is near to that of the optimal policy. 

Before proceeding, we prove the following important lemma which will be used several times later.

\begin{lemma}\label{Lemma:baseLemma}

For any Hermitian Symmetric matrix $\mathbf{W}\in \mathbb{C}^{N \times N}$, and $\mathbf{v}\in \mathbb{C}^{N \times 1}$ we have, 
$$Tr(\mathbf{W}\mathbf{v}\mathbf{v}^\ast) \leq \lambda_{max}||\mathbf{v}||^2,$$
and the equality holds when  $\mathbf{v}=||\mathbf{v}||\mathbf{u}_{max}$, where $\mathbf{u}_{max}$ is the eigenvector corresponding to the maximum eigenvalue of $\mathbf{W}$.
\end{lemma}
\begin{proof}
Since $\mathbf{W}$ is Hermitian symmetric, it can be written as  $\mathbf{W}=\mathbf{U}^\ast\mathbf{\Lambda}\mathbf{U}$, where $\mathbf{U}$  and $\mathbf{\Lambda}$ are some unitary and   diagonal matrices, respectively. Consequently, we have,
\begin{align*}
Tr(\mathbf{W}\mathbf{v}\mathbf{v}^\ast) &=Tr(\mathbf{v}^\ast\mathbf{U}^\ast\mathbf{\Lambda}\mathbf{U}\mathbf{})  \\
&=Tr(\mathbf{z}^\ast\mathbf{\Lambda}\mathbf{z})=\sum_{i=1}^N \lambda_i|z_i|^2   \\
&\le\lambda_{max}||\mathbf{v}||^2.  
\end{align*}
Note that the  inequality above is tight and the equality can be achieved by adopting $\mathbf{v}$ in direction of $\mathbf{u}_{max}$, i.e., $\mathbf{v} =||\mathbf{v}||\mathbf{u}_{max} $ as a result $\mathbf{z}$ will have only one none zero element equal to $||\mathbf{v}||^2$ at index that  corresponds to $\lambda_{max}$.
\end{proof}

\subsection{{\color{black}Optimal Transmission Policy for Single Receiver}}

The following theorem states the optimal policy for the special case of single E-R.

\begin{theorem}
\label{theo:MinTran}
The optimal solution to the problem of equation \eqref{P1} in the case of one E-R is as follows:
\begin{align}
\mathbf{x}[l] =  \left\{
\begin{array}{cc}
	P_{peak}\mathbf{u}_{max}[l],\ \   &\lambda_{max}[l]\ge \lambda_{Th} \\
0, & otherwise,
\end{array} \right.
\label{equ:garant2}
\end{align}
where $\mathbf{u}_{max}[l]$ is the eigenvector of matrix $\mathbf{W}_1[l]$ with maximum eigenvalue ($\lambda_{max}[l]$) in timeslot $l$. $\lambda_{Th}$ must satisfy the following equality:
\begin{equation} \label{equ:mintrans}
\int_{\lambda_{Th}}^\infty \alpha f_{\lambda_{max}}(\alpha)d\alpha=\frac{P_1^{recv}}{P_{peak}}, 
\end{equation}
where $f_{\lambda_{max}}(\alpha)$ denotes the probability distribution function (pdf) of $\lambda_{max}$. As $\lambda_{max}[l]$ is a function of $\mathbf{H}_i[l]$, its distribution in any slot is the same as other slots and we can drop the index for ease of notation.
\end{theorem}
\begin{proof}
See Appendix \ref{app:MinTran}.
\end{proof}

Note that Theorem \ref{theo:MinTran} introduces a two level policy as the optimal one. Under this policy,  when channel is in good condition, the E-AP should transmit with maximum power; otherwise, it should stop transmission. The channel condition is determined in terms of the gain of the best path toward the E-R, and is compared to a threshold value determined in \eqref{equ:mintrans}. If it is lower than the threshold, then the E-AP will not transmit on that timeslot. This policy seems reasonable, since avoiding transmission in bad channel conditions and saving energy for transmission in better channel conditions increases the efficiency.  The value of $\lambda_{Th}$ can be calculated by solving Eq. \eqref{equ:mintrans} with a simple line search method which requires the knowledge of the distribution of $\lambda_{max}$. 

\subsection{{\color{black}Transmission Policy for Energy-Limited Case}}
It should be noted that solving the formulated problem in  \ref{P1} is not trivial since it is a non-convex optimization problem,  and also its solution requires having knowledge of channel states distribution, which is not available in many practical cases. In this sub-section, to address the aforementioned  challenges and solve this problem, we  propose a policy which does not require the channel states distribution, using a stochastic approach based on Min Drift Plus Penalty (MDPP) algorithm \cite{Neely2010}.  We will then show that the time-averaged transmit power under the proposed policy, denoted by   $\bar{Q}^{tran}_{MDPP}$, is  close to the one under the optimal policy, denoted by, $\bar{Q}^{tran}_{opt}$.

The pseudo-code of the proposed transmission policy is presented in Algorithm \ref{alg:TranMin}. The  proposed policy follows a similar  two-level transmission strategy as in the optimal solution for the single E-R case. 
The process $Z_i,\ \forall i = 1,2,\ldots,K$ introduced in Algorithm \ref{alg:TranMin} represents a virtual queue that captures the deviation of the average received power of the E-Rs from their minimum requirement denoted in constraint \eqref{equ:P1ExpConst}. The E-AP steers its beam toward E-Rs with larger queue backlog and/or better channel condition,  as a result the E-Rs which have received less power in previous timeslots have higher priority for receiving  power. The parameter $V$  in Algorithm \ref{alg:TranMin} is a control variable which maintains the trade-off between minimization of average transmitted power of the E-AP and the duration of the short-time deviations from the minimum power requirements of the E-Rs.
Under this algorithm, the E-AP estimates the channel at the beginning of $l$'th  timeslot, and calculates the weighted sum  matrix $\mathbf{W'}[l]$, as in step \ref{step:sumchannel}\footnote{Throughout this paper, we let $\mathbf{I}_m$ denote the $m\times m$ identity matrix.}. Then if the largest eigenvalue of $\mathbf{W'}[l]$, denoted by $\lambda^\mathbf{W'}_{max}[l]$ is greater than zero, the E-AP will transmit with its maximum power in direction of the $\mathbf{u}_{max}^{\mathbf{W'}}[l]$ that is the eigenvector with largest eigenvalue. Otherwise, the E-AP will not transmit any power.  The virtual queues will be updated at the end of timeslot. The following theorem  characterizes the optimality gap of Algorithm \ref{alg:TranMin}. 
\begin{theorem} \label{theo:TranMin}
The  E-AP transmission  policy described by  Algorithm \ref{alg:TranMin}
\begin{enumerate}[label=(\alph*)]
\item
Satisfies constraints \eqref{equ:P1ExpConst}.
\item
Yields an average transmitted power within a maximum distance of $\frac{B}{V}$ from the optimal solution: ($\bar{Q}^{tran}_{opt}\le\bar{Q}^{tran}_{MDPP}\le \bar{Q}^{tran}_{opt}+\frac{B}{V}$), where $B = \frac{K}{2}P_{peak}^2$. 
\end{enumerate}
\end{theorem}
\begin{algorithm}[t]
\caption{E-AP algorithm in energy-limited case}\label{alg:TranMin}
\begin{algorithmic}[1]
\State \textbf{Initialization:} $l\gets 0, Z_i[0] \gets 0,\; \forall i = 1,2,\ldots,K$.
\While{(true)}
\State Estimate $\mathbf{H}_i[l],\; \forall i = 1,2,\ldots,K.$
\State $\mathbf{W}_i[l] \gets \mathbf{H}^\ast_i[l]\mathbf{H}_i[l],\; \forall i = 1,2,\ldots,K$.
\State $\mathbf{W'}[l] \gets \sum_{i=1}^K Z_i[l]\mathbf{W}_i[l]-V\mathbf{I}_N$.\label{step:sumchannel}
\If{$\lambda_{max}^{{\mathbf{W'}}}[l]>0$}
\State $\mathbf{x}[l] \gets P_{peak}\mathbf{u}_{max}^{\mathbf{W'}}[l]$,
\Else
\State $\mathbf{x}[l] \gets 0$.
\EndIf
\State $Z_i[l+1] \gets max\{Z_i[l]+P_i-Tr(\mathbf{W}_i[l]\mathbf{x}[l]\mathbf{x}^\ast[l]),0\},\; \forall  i = 1,2,\ldots,K$.
\State $l\gets l+1$.
\EndWhile
\end{algorithmic}
\end{algorithm}

\begin{proof}

{\color{black} Consider the following definitions for quadratic Lyapunov function and Lyapunov drift, respectively, 
\begin{equation}
\mathbb{L}(\mathbf{Z}[l]) \triangleq \frac{1}{2}\sum_{i=1}^K Z_i^2[l],
\end{equation}
\begin{equation}
\Delta(\mathbf{Z}[l]) \triangleq \mathbb{E}[(\mathbb{L}(\mathbf{Z}[l+1])-\mathbb{L}(\mathbf{Z}[l]))|\mathbf{Z}[l]],
\end{equation}
where $\mathbf{Z}[l]\triangleq({Z}_1[l],{Z}_2[l],...,{Z}_K[l])$. Let us define the drift-plus-penalty function as, 
\begin{equation}\label{equ:DPP}
\Delta(\mathbf{Z}[l]) + V\mathbb{E}[Q_{AP}[l]|\mathbf{Z}[l]], 
\end{equation}  
 where $Q_{AP}[l]=Tr(\mathbf{x}[l]\mathbf{x}^\ast[l])$.  The first term in Eq. \eqref{equ:DPP} is a measure of the expected total backlog increment in all virtual queue
and the second term is the expected  transmitted power from the E-AP, where both are condition on the current queue backlog. The intuition behind MDPP technique is that a proper policy minimizes this function and obtains a balance between  transmitted power and virtual queues' backlog. The following lemma establishes an upper-bound on the drift-plus-penalty function.
\begin{lemma} \label{Lem:DriftMinTran}
 The drift plus penalty function has the following upper bound:
\begin{equation}\label{equ:LyapunovIneq}
\Delta(\mathbf{Z}[l]) + V\mathbb{E}[Q_{AP}[l]|\mathbf{Z}[l]] \le B + V\mathbb{E}[Q_{AP}[l]|\mathbf{Z}[l]] + \sum_{i=1}^K Z_i[l]\mathbb{E}[  Q^{d}_i [l] |\mathbf{Z}[l]], 
\end{equation}
where $Q^{d}_i [l]=  P_i-Tr(\mathbf{W}_i[l]\mathbf{x}[l]\mathbf{x}^\ast[l])$ and $B = {K\over 2} P_{peak}^2$. 
\end{lemma} 
\begin{proof}\renewcommand{\qedsymbol}{}
See Appendix \ref{app:Lemma2}. 
\end{proof}
Using Lyapunov optimization theorem, it is shown in \cite{Neely2010} that a policy which minimizes the Right Hand Side (RHS) of \eqref{equ:LyapunovIneq} at each timeslot has the claimed properties (a) and (b) in Theorem  \ref{theo:TranMin}. Hence in order to prove the theorem we only need to show that Algorithm  \ref{alg:TranMin}  minimizes the RHS of  \eqref{equ:LyapunovIneq} over all possible policies.
To that end, in each timeslot $l$, the E-AP observes the queue backlogs and chooses $\mathbf{x}[l]$ equal to the solution of the following optimization problem:
\begin{mini!}|l|
 {\mathbf{x}[l] }{ V{Q_{AP}}[l]+\sum_{i=1}^K Z_i[l]Q^d_i[l]}{\label{equ:MDPPMinTran}}{}
  \addConstraint{  Tr(\mathbf{x}[l]\mathbf{x}^\ast[l])\le P_{peak}}. 
 \end{mini!}

Substituting $Q^d_i, \ i = 0,\ldots,K$ in equation \eqref{equ:MDPPMinTran} and neglecting the constant terms, we can rewrite the optimization problem as,
\begin{maxi!}|l|
 {\mathbf{x}[l] }{Tr(\mathbf{W}'[l]\mathbf{x}[l]\mathbf{x}^\ast[l])}{\label{equ:MDPPMinTran2}}{}
  \addConstraint{  Tr(\mathbf{x}[l]\mathbf{x}^\ast[l])\le P_{peak}}. 
 \end{maxi!}
where $\mathbf{W'}[l] \triangleq \sum_{i=1}^K Z_i[l]\mathbf{W}_i[l]-V\mathbf{I}_N$. Now, using Lemma \ref{Lemma:baseLemma}, the solution of the optimization problem \eqref{equ:MDPPMinTran} is obtained as follows: 

 
\begin{align}
\mathbf{x}[l] =  \left\{
\begin{array}{cc}
P_{peak}\mathbf{u}_{max}^{\mathbf{W'}}[l],\ \   &\lambda_{max}^{\mathbf{W'}}[l]\ge 0, \\
0, & otherwise,
\end{array} \right.
\end{align}
This is exactly the same as the policy presented in Algorithm 1. Therefore, this algorithm minimizes the RHS of  \eqref{equ:LyapunovIneq}, and hence, according to Lyapunov optimization theorem \cite{Neely2010} satisfies parts (a) and (b) in Theorem \ref{theo:TranMin}.


}
\end{proof}

\section{Maximizing Average Received Energy in power-limited Case}
In this section, we consider an E-AP that is connected to a stable power source, and multiple E-Rs that receive energy from this E-AP\footnote{We borrow our notation from the previous section in the current and following sections. The references of the notations are clear from the context of each section.}. A real-world application of this scenario is  wireless charging of  battery-powered devices in smart homes. Since wireless  chargers are plugged into power outlets, there is no limitation on their available energy, but the input power is limited. An optimal transmitting policy of the E-AP aims at maximixing the  power transmission efficiency by maximizing the total received power of the E-Rs. 
 Consequently, the optimal policy can be formulates as a solution of the following  problem:
\begin{maxi!}|l|
 {\{\mathbf{x}(\mathbf{H})\}}{\lim_{L \rightarrow \infty}\frac{1}{L}\sum_{l=0}^{L-1}\sum_{i=1}^K\mathbb{E}[Tr(\mathbf{W}_i[l]\mathbf{x}[l]\mathbf{x}^\ast[l])]}{\label{P2}}{}\label{P2Obj}
  \addConstraint{ \lim_{L \rightarrow \infty}\frac{1}{L}\sum_{l=0}^{L-1}\mathbb{E}[Tr(\mathbf{x}[l]\mathbf{x}^\ast[l])]\le P_{avg}} \label{equ:MAXRECVPTRANAVG}
\addConstraint{{Tr(\mathbf{x}[l]\mathbf{x}^\ast[l])}\le P_{peak}, \quad \forall l \ge 0 }, \label{equ:P2peak}
 \end{maxi!}
where constraint \eqref{equ:MAXRECVPTRANAVG} guarantees that the  average transmission power does not exceed $P_{avg}$, and  constraint \eqref{equ:P2peak} is the physical   limitation on the instantaneous transmit power of the E-AP. Solving Problem \eqref{P2} involves challenges similar to Problem \eqref{P1}, i.e., the problem is non-convex due to constraint \eqref{equ:MAXRECVPTRANAVG} and there is no closed form expression for terms with time-averaged expectations. 
In a similar vein as the energy-limited case, in the sequel, we first assume that the channel statistics are available and obtain the optimal solution to problem \eqref{P2} to find the optimal policy. Then, based on MDPP, we propose a near-optimal transmission policy that does not require the channel statistics, and  derive the optimality gap of its performance as well.


\subsection{Optimal Policy}
%

The following theorem derives an optimal solution for Problem \eqref{P2}. 
\begin{theorem}\label{theo:MaxRecv}
The following transmission policy maximizes \eqref{P2Obj} and satisfies constraints \eqref{equ:MAXRECVPTRANAVG} and \eqref{equ:P2peak}.
At each timeslot the E-AP  estimates the channel  and determines the beamforming vector as:
\begin{align}
\mathbf{x}[l] =  \left\{
\begin{array}{cc}
P_{peak}\mathbf{u}_{max}^{\mathbf{W'}}[l],\ \   &\lambda_{max}^{\mathbf{W'}}[l]\ge \lambda_{Th}^{\mathbf{W'}}, \\
0, & otherwise,
\end{array} \right.
\label{equ:MDPPMaxRecv}
\end{align}
where $\mathbf{u}_{max}^{\mathbf{W'}}$ is the eigenvector of matrix $\mathbf{W'}[l] \triangleq \sum_{i=1}^K \mathbf{W}_i[l]$ associated with the largest eigenvalue ($\lambda_{max}^{\mathbf{W'}}[l]$) and
\begin{equation}
\lambda_{Th}^{\mathbf{W'}}=F^{-1}_{\lambda_{max}^{\mathbf{W'}}}(1-\frac{P_{avg}}{P_{peak}}), \label{equ:lambaThTheorem3}
\end{equation} 
where $F^{-1}_{\lambda_{max}^{\mathbf{W'}}}$ is the inverse cumulative distribution function of $\lambda_{max}^{\mathbf{W'}}$.  
\end{theorem}
\begin{proof}\renewcommand{\qedsymbol}{}
See Appendix \ref{app:OptRecvMax}
\end{proof}

Note that the transmission policy introduced in \eqref{equ:MDPPMaxRecv} concentrates the transmission beam toward a virtual E-R with a channel matrix equal to the sum of all channel matrices. Under this policy the beam is always biased toward the E-Rs with better channel conditions. Moreover, to calculate the optimal threshold in  \eqref{equ:lambaThTheorem3}, the E-AP  needs to know the  distribution of the largest eigenvalue of the  sum of channel matrices, which may not be available in general. This issue makes the optimal policy impractical in many applications. Nevertheless, still he optimal solution can serve as an upper-bound for the performance of any other policy, and also  sheds a light on  the structure of a proper transmission strategy. 
\subsection{ Transmission Policy for Power-Limited Case}
In this subsection, we propose a transmission policy  for power-limited case, based on MDPP technique. As discussed before, this technique only needs the instantaneous CSI and adapts to variation in channel distribution. This policy is introduced in Algorithm \ref{alg:RecvMax}. The virtual queue $Z_1$ in this algorithm captures the deviation of the average transmitted power from $P_{avg}$. 
The beamforming vector is  determined in steps    \ref{alg:s1}  to  \ref{alg:s2}  of  Algorithm  \ref{alg:RecvMax}.
Similar to the optimal solution in Theorem \ref{theo:MaxRecv}, the beamforming vector in Algorithm  \ref{alg:RecvMax} is determined by $\lambda_{max}^{\mathbf{W}'}$, which is the eigenvector of the sum channel 
matrix $\mathbf{W}'$ associated with the  largest eigenvalue. The E-AP updates $Z_1$ at the end of each timeslot.
The following theorem states that under the proposed policy, the  time averaged expected total received power, $\bar{Q}_{PL}^{MDPP}$, is within a bounded distance of the one under optimal policy, $\bar{Q}_{PL}^{Opt}$. 

\begin{theorem}\label{theo:RecvMax}
The E-AP transmission policy given in Algorithm \ref{alg:RecvMax}:
\begin{enumerate}[label=(\alph*)]
\item
Satisfies the constraint  \eqref{equ:MAXRECVPTRANAVG}.
\item
Yields an average received power within a maximum distance of $\frac{B}{V}$ from the optimal solution, i.e.,  $\bar{Q}_{PL}^{Opt} \le \bar{Q}_{PL}^{MDPP}\le \bar{Q}_{PL}^{opt}+\frac{B}{V}$, where $B = \frac{1}{2}P_{peak}^2$ and $V$ is a control parameter of MDPP algorithm. 
\end{enumerate}
\end{theorem}

\begin{algorithm}[t]
\caption{E-AP algorithm in power-limited case}\label{alg:RecvMax}
\begin{algorithmic}[1]
\State \textbf{Initialization:} $l\gets 0, Z_1[0] \gets 0$.
\While{(true)}
\State Estimate $\mathbf{H_i},\;\forall i = 1,2,\ldots,K $.
\State $\mathbf{W}_i[l] \gets \mathbf{H}^\ast_i[l]\mathbf{H}_i[l],\;\forall i = 1,2,\ldots,K$.
\State $\mathbf{W'}[l] \gets V\sum_{i=1}^K \mathbf{W}_i[l]-Z_1[l]\mathbf{I}$. \label{alg:s1}
\If{$\lambda_{max}^{{\mathbf{W}'}}[l]>0$}
\State $\mathbf{x}[l]\gets P_{peak}\mathbf{u}_{max}^{\mathbf{W'}}[l]$,
\Else
\State $\mathbf{x}[l] \gets 0$.
\EndIf \label{alg:s2}
\State $Z_1[l+1] \gets max\{Z_1[l]+Tr(\mathbf{x}[l]\mathbf{x}^\ast[l])-P_{avg},0\}$.
\State $l \gets l+1$.
\EndWhile
\end{algorithmic}
\end{algorithm}

The proof is similar to the proof of Theorem \ref{theo:TranMin}, and is omitted here for brevity.

\section{Considering Fairness in Maximizing Received Energy }
The proposed transmission policy in Algorithm \ref{alg:RecvMax} is highly biased in flavor of those E-Rs that are nearer to the E-AP. This is because the nearer E-Rs receive more energy  than the farther E-Rs if the same amount of energy is transmitted toward them. In this section, we aim to ensure fairness in designing transmission policies. For this purpose, we investigate  two techniques for imposing fairness among the E-Rs, namely Max-Min Fairness (MMF) and QoS-aware proportional fairness (QPF). In the following two subsections, we introduce each technique and propose transmission policies for them.

%
%

\subsection{Max-Min Fairness}
A common technique for achieving fairness among E-Rs is maximizing the minimum of the average received powers of different E-Rs. This is known as max-min fairness (MMF) \cite{Marbach2002}, which results in
a balance between the received power of the E-Rs, but with an expense of decreasing the total received power of the E-Rs. Accordingly, the MMF policy can be formulated as the solution of the following problem:
\begin{maxi!}|l|
 {\{\mathbf{x}(\mathbf{H})\}}{\bar{Q}_{min}\triangleq\min_i \lim_{L \rightarrow \infty}\frac{1}{L}\sum_{l=0}^{L-1} \mathbb{E}[Tr(\mathbf{W}_i[l]\mathbf{x}[l]\mathbf{x}^\ast[l])]}{\label{P3}}{} \label{equ:P3Obj}
  \addConstraint{ \lim_{L \rightarrow \infty}\frac{1}{L}\sum_{l=0}^{L-1}\mathbb{E}[Tr(\mathbf{x}[l]\mathbf{x}^\ast[l])]\le P_{avg}} \label{equ:PavgConstMinMaxFair}
\addConstraint{{Tr(\mathbf{x}[l]\mathbf{x}^\ast[l])}\le P_{peak}, \quad \forall l  \ge 0 }.
 \end{maxi!}
%
%

Same as before,  we avoid struggling  with the  non-convex problem \eqref{P3} by introducing a policy  and analyzing its performance.   Algorithm \ref{alg:RecvMinMaxFairness} describes the proposed MMF policy, wich is based on the MDPP technique for maximizing some concave function of time averages. Same as before, the proposed MMF follows a   two-level structure, and focuses the transmission beam  toward a virtual E-R. In this policy,  the channel matrix of the virtual E-R is a weighted sum  of the channel matrices of all E-Rs, and the  weights are determined by the virtual queues  ${G}_i,\; i=1,\ldots,K$. The backlog of these virtual queues grows faster for the E-Rs which receive less power. As a consequent,  these E-Rs have a greater weight in the weighted sum channel matrix, and  will receive more power in the consequent timeslots. Let  $\bar{Q}_{min}^{MMF}$ and $\bar{Q}_{min}^{Opt}$ denote  the minimum time-averaged received power under the MMF policy and the optimal policy, respectively. The following theorem discusses the optimality of the MMF policy. 
\begin{theorem}\label{theo:MMFPolicy}
 The MMF policy for the E-AP transmission, described by Algorithm \ref{alg:RecvMinMaxFairness}:
\begin{enumerate}[label=(\alph*)]
\item
Satisfies the constraints \eqref{equ:PavgConstMinMaxFair}.
\item
Yields a minimum average received power that is within a maximum distance of $\frac{B}{V}$ from the optimal solution, i.e., $\bar{Q}_{min}^{Opt}-\frac{B}{V}\le \bar{Q}_{min}^{MMF} \le \bar{Q}_{min}^{Opt}$, where $B = \frac{K+1}{2}P_{peak}^2$ and $V$ is a control parameter of the MDPP algorithm. 
\end{enumerate}

\end{theorem}
\begin{algorithm}[t]
\caption{E-AP algorithm in power-limited case considering Max-Min fairness}\label{alg:RecvMinMaxFairness}
\begin{algorithmic}[1]
\State \textbf{Initialization:} $l\gets 0, Z_1[0] \gets 0, G_i[0] \gets 0,\;\forall i = 1,2,\ldots,K $.
\While{(true)}
\State Estimate $\mathbf{H}_i,\; \forall i = 1,2,\ldots,K$.
\State $\mathbf{W}_i[l] \gets \mathbf{H}^\ast_i[l]\mathbf{H}_i[l],\;\forall i = 1,2,\ldots,K$.
\If{$V>\sum_{i=1}^K G_i[l]$}
\State $\gamma_i[l]\gets P_{peak},\;\forall i = 1,2,\ldots,K$,
\Else
\State $\gamma_i[l]\gets 0,\;\forall i = 1,2,\ldots,K$.
\EndIf
\State $\mathbf{W'}[l] \gets \sum_{i=1}^K {G}_i[l]\mathbf{W}_i[l]-Z_1[l]\mathbf{I}$.
\If{$\lambda_{max}^{{\mathbf{W}'}}[l]>0$}
\State $\mathbf{x}[l] \gets P_{peak}\mathbf{u}_{max}^{{\mathbf{W}'}}$,
\Else
\State $\mathbf{x}[l] \gets 0$.
\EndIf
\State $Z_1[l+1] \gets \max\{Z_1[l]+Tr(\mathbf{x}[l]\mathbf{x}^\ast[l])-P_{avg},0\}$.
\State $G_i[l+1]\gets \max \{G_i[l]+\gamma_i[l]-Tr(W_i[l]\mathbf{x}[l]\mathbf{x}^\ast[l])\},\ \forall  i = 1,2,...,K$.
\State $l\gets l+1$.
\EndWhile
\end{algorithmic}
\end{algorithm}
\begin{proof}
The analysis in Theorem \ref{theo:TranMin} is not directly applicable here, since the objective function in \eqref{equ:P3Obj} is a function of a time-average, rather than a time-averaged quantity.
To prove the above theorem, using a similar approach as in \cite{Neely2010}, we introduce the auxiliary variables $\bm{\gamma}[l]=(\gamma_1[l],...,\gamma_K[l])$ and define the a modified optimisaiton problem as follows:
\begin{maxi!}|l|
 {\{\bm{x}(\mathbf{H}),\bm{\gamma}(\mathbf{H})\}}{\overline{\phi(\bm{\gamma})}}{\label{P3'}}{} \label{equ:objectiveModif}
  \addConstraint{ \lim_{L \rightarrow \infty}\frac{1}{L}\sum_{l=0}^{L-1}\mathbb{E}[Tr(\mathbf{x}[l]\mathbf{x}^\ast[l])]\le P_{avg}} \label{equ:ModifiedProblemPavg} 
\addConstraint{\bar{\gamma}_i\le \bar{Q}_i,\ \forall k \in \{1,...,K\}} \label{equ:ModifiedProblemConst} 
\addConstraint{{Tr(\mathbf{x}[l]\mathbf{x}^\ast[l])}\le P_{peak}, \forall l \ge 0 }, \label{equ:ModifiedProblemPpeak} 
 \end{maxi!}
where 
\begin{align*}
\phi(r_1,r_2,...,r_K)&\triangleq\min_{i\in\{1,\ldots,K\}} r_i \\
\overline{\phi(\bm{\gamma})}&\triangleq\lim_{L \rightarrow \infty}\frac{1}{L}\sum_{l=0}^{L-1} \mathbb{E}[ \phi(\gamma_1,\gamma_2,...,\gamma_K)   ],\\ 
 \bar{\gamma}_i&\triangleq\lim_{L \rightarrow \infty}\frac{1}{L}\sum_{l=0}^{L-1}\mathbb{E}[\gamma_i[l]], \\ 
\bar{Q}_i &\triangleq \lim_{L \rightarrow \infty}\frac{1}{L}\sum_{l=0}^{L-1}\mathbb{E}[Tr(\mathbf{W}_i[l]\mathbf{S}[l])].
\end{align*}

Now, we  first prove that the optimal transmission policy of the modified problem is also the optimal policy for the original problem. Then, using this result and noting that the modified problem is in the form of the problems we encountered before, the intended properties (a) and (b) can be proved using a similar approach to the proof of Theorem \ref{theo:TranMin}. The proof of the later is omitted for brevity, to prove the former, we first show that the optimal transmission policy  of the original problem, $\bm{x}^{opt}(\bm{H}),$ is a feasible policy for the modified problem. This can be achieved by choosing $\bm{\gamma}[l] = \bm{\gamma}^{opt}=({\bar{Q}^{opt}_1},...,\bar{Q}^{opt}_K),\ \forall l$, where $\bar{Q}^{opt}_i$ is the time averaged received power to receiver $i$ under $\bm{x}^{opt}(\bm{H})$. One can verify that $ (\bm{x}^{opt}(\bm{H}),\bm{\gamma}^{opt})$ satisfies  the constraints \eqref{equ:ModifiedProblemPavg}-\eqref{equ:ModifiedProblemPpeak}. Furthermore, by this choice of the arguments, the value of the objective function  \eqref{equ:objectiveModif}  equals $\phi_{opt}\triangleq\phi({\bar{Q}^{opt}_1},...,\bar{Q}^{opt}_m)$, which is the maximum value of the objective function in Problem \eqref{P3}. Therefore, we have $\overline{\phi(\bm{\gamma}^\ast)}\ge \phi_{opt}$, where $\bm{\gamma}^\ast$ is the maximizer of of Problem \eqref{P3'}.

%
It is straightforward to verify that the $\phi: \mathbb{R}^K\to \mathbb{R}$ is continuous, concave and entrywise non-decreasing. Hence, we can write
\begin{equation}
\phi({\bar{Q}^{\ast}_1},...,\bar{Q}^{\ast}_K)\stackrel{a}{\ge}\phi(\bar{\bm{\gamma}}^\ast)\stackrel{b}{\ge} \overline{\phi(\bm{\gamma}^\ast)} \ge \phi_{opt} \label{equ:ModifiedVsOriginal},
\end{equation}
where $\bar{Q}^{\ast}_i$ is the time averaged received power to receiver $i$ under the maximizing solution of  Problem \eqref{P3'}. Inequality $(a)$ follows from equation \eqref{equ:ModifiedProblemConst} and non-decreasing entrywise property of $\phi$ and $(b)$ follows from Jensen's inequality. On the other hand, we have, 
\begin{equation}
 \phi_{opt} \ge \phi({\bar{Q}^{\ast}_1},...,\bar{Q}^{\ast}_K),
\label{equ:intarafi}
\end{equation}
the above inequality holds since all the feasible transmission polices of Problem \eqref{P3'}, also satisfy the constraints of Problem \eqref{P3}. From \eqref{equ:ModifiedVsOriginal} and \eqref{equ:intarafi} we conclude that,
\begin{equation}
 \phi_{opt} = \phi({\bar{Q}^{\ast}_1},...,\bar{Q}^{\ast}_K),
\end{equation}
which is the desired result.
\end{proof}

\subsection{QoS-aware Proportional Fairness}

Maximizing the minimum received power is a strict policy, in the sense that it focuses most of the transmitted power toward farther E-Rs, no matter how far they are from the E-AP. As such, this policy results in drastic degradation in total received power if some of the E-R are very far from the E-AP.  To address this issue in considering fairness, the proportional fairness technique can be used, which  makes a trade-off between maximizing the total received power and the fairness. A proportional fair policy aims to maximize sum of the logarithms of the received power in E-Rs, as the objective utility. Due to the specific structure of a logarithmic function, increasing the received power in the farther E-Rs leads to a greater increase in total utility, but at the same time, increases the power consumption of the E-AP. An optimal policy balances this trade off, and maximizes the total utility. The formal definition of proportional fairness is presented in \cite{Kelly1998}. The QPF policy maximizes the sum of the logarithm of the received powers of the E-Rs, while providing the required QoS of each E-R, i.e., the minimum received power requirement of each E-R.  To find the QPF policy, we  formulate the following optimization problem:
\begin{maxi!}|l|
 {\{\mathbf{x}(\mathbf{H})\}}{\bar{Q}_{LogT}\triangleq \sum_{i=1}^K \log(\bar{Q}_i)}{\label{P4}}{}
  \addConstraint{ \lim_{L \rightarrow \infty}\frac{1}{L}\sum_{l=0}^{L-1}\mathbb{E}[Tr(\bm{W}_i[l]\mathbf{x}[l]\mathbf{x}^\ast[l])]\ge P_{min}} 
\addConstraint{ \lim_{L \rightarrow \infty}\frac{1}{L}\sum_{l=0}^{L-1}\mathbb{E}[Tr(\mathbf{x}[l]\mathbf{x}^\ast[l])]\le P_{avg}} \label{equ:PavgConstPropFair}
\addConstraint{{Tr(\mathbf{x}[l]\mathbf{x}^\ast[l])}\le P_{peak}, \forall l \ge 0 }\ .
 \end{maxi!}

Algorithm \ref{alg:RecvPropFairness} solves the formulated problem and describes the proposed QPF policy. The performance of this algorithm is analyzed in the following theorem.
%
%
\begin{theorem}
Algorithm \ref{alg:RecvPropFairness} describes the QPF policy for the E-AP transmission. This policy 
\begin{itemize}
\item
Satisfies the constraints of \eqref{equ:PavgConstPropFair}.
\item
Yields an objective funtion that is within a maximum distance of $\frac{B}{V}$ from the optimal solution: ($\bar{Q}_{LogT}^{opt}-\frac{B}{V}\le \bar{Q}_{LogT}^{MDPP} \le \bar{Q}_{LogT}^{opt}$), where $B = \frac{2K+1}{2}P_{peak}^2$ and $V$ is a parameter of the MDPP algorithm. 
\end{itemize}
\end{theorem}
\begin{algorithm}[t]
\caption{E-AP algorithm in power-limited case considering QoS-aware Proportional fairness}\label{alg:RecvPropFairness}
\begin{algorithmic}[1]
\State \textbf{Initialization:} $l \gets 0, Z_i[0] \gets 0,\ \forall i = 1,2,...,K+1, G_m[0] \gets 0, \forall m = 1,2,...,K$.
\While{(true)}
\State Estimate $\mathbf{H_i},\;\forall i = 1,2,\ldots,K$.
\State $\mathbf{W}_i[l] \gets \mathbf{H}^\ast_i[l]\mathbf{H}_i[l],\;\forall i = 1,2,\ldots,K$.
\State $\gamma_i[l] \gets \min\{\frac{V}{G_i[l]}, P_{peak}\},\; \forall i = 1,2,\ldots,K$.
\State $\mathbf{W'}[l] \gets \sum_{i=1}^K (Z_i[l]+{G_i}[l])\mathbf{W}_i[l]-Z_{K+1}[l]\mathbf{I}$.
\If{$\lambda_{max}^{{\mathbf{W}'}}[l]>0$}
\State $\mathbf{x}[l] \gets P_{peak}\mathbf{u}_{max}^{{\mathbf{W}'}}$,
\Else
\State $\mathbf{x}[l] \gets 0$.
\EndIf
\State $G_i[l+1]\gets \max \{G_i[l]+\gamma_i[l]-Tr(\mathbf{W}_i[l]\mathbf{x}[l]\mathbf{x}^\ast[l]),0\},\;\forall  i = 1,2,...,K$.
\State $Z_i[l+1] \gets \max\{Z_i[l]+P_{min}-Tr(\mathbf{W}_i[l]\mathbf{x}[l]\mathbf{x}^\ast[l]),0\},\;\forall i = 1,2,...,K$.
\State $Z_{K+1}[l+1] \gets \max\{Z_{K+1}[l]+Tr(\mathbf{x}[l]\mathbf{x}^\ast[l])-P_{avg},0\}$.
\State $l\gets l+1$.
\EndWhile
\end{algorithmic}
\end{algorithm}

The proof is similar to the proof of Theorem \ref{theo:MMFPolicy}, and is omitted here for brevity.

\section{Numerical Results}
In this section, we evaluate the performance of the proposed algorithms with various numerical results. Unless noted otherwise, we consider an E-AP and two E-Rs (each equipped with four receive antennas). The considered network topology is shown in Fig. {\ref{fig:SimSys}}. All the proposed algorithms  are run for $10^5$ timeslots. All the simulation results have been obtained by MATLAB R2015b on a simulation platform with a Windows server 2008, Intel Xeon E5-2650v3 CPU (2.3GHz), and 64GB RAM. 

\begin{figure}
\centering
\resizebox{0.4\textwidth}{!}{
\begin{tikzpicture}[scale=0.6]
   [WD/.style={fill=blue!20,thick},
   HAP/.style={fill=black!20,thick}]

   \node(HAP) at ( 0,0) [minimum size = 5mm,,label=left: E-AP,label=below:$(0\;0)$] {\includegraphics[width=.06\textwidth]{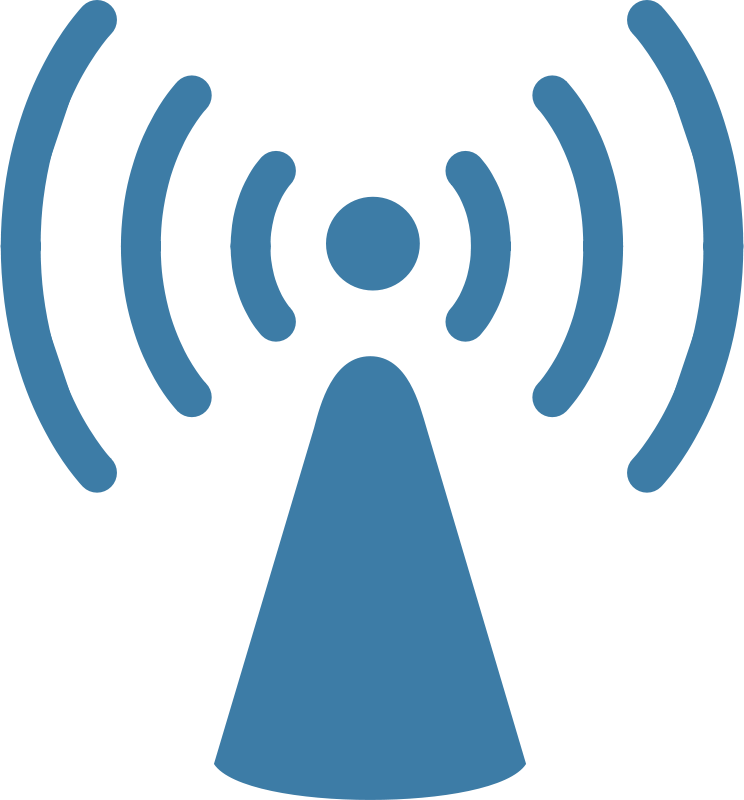}};
   \node(WD1) at ( 3,3) [minimum size = 5mm,,label=above:E-R1,label=right:$(0.3 \;0.3)$ ] {\includegraphics[width=.06\textwidth]{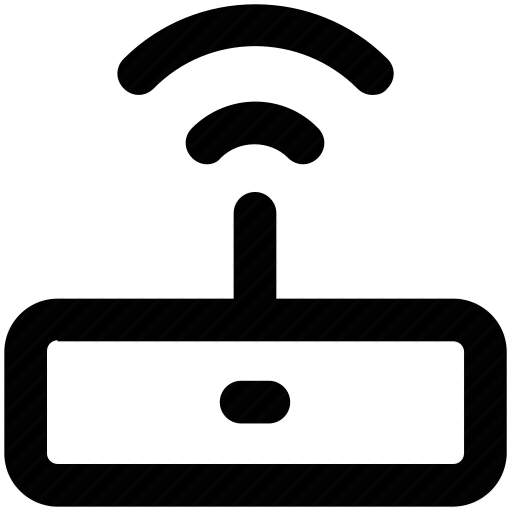}};
   \node(WD2) at (0,7.0711) [minimum size = 12mm,,label=above:E-R2 ,label=left:$(0 \;\; 0.5\sqrt{2})$] {\includegraphics[width=.06\textwidth]{WD.png}};
   \draw [->, line width = 0.6mm, color = red] ([xshift=0.1ex]HAP.north east) to ([xshift=-0.1ex]WD1.south west);
   \draw [->, line width = 0.6mm, color = red] ([yshift=0.1ex]HAP.north) to ([yshift=-0.15ex]WD2.south);
\draw (0,0) circle (4cm);
\draw (0,0) circle (6.5cm);
\end{tikzpicture} }
\caption{The considered network topology }
\label{fig:SimSys}
\end{figure}

\subsection{Energy-Limited Case}
As mentioned before, in the energy-limited case, the E-AP has a limited battery and hence, targets at transferring only the required power level of the E-Rs. 
Fig. \ref{fig:MinTranTwoER} shows the average transmitted power of the E-AP under the proposed policy in Algorithm \ref{alg:TranMin} versus the number of the E-AP's antennas. The maximum transmit power of the E-AP is considered to be 5 W, and three different values of  5, 10, and 15 mW are considered for the required power level of the E-Rs. As can be verified from the figure, under the proposed Algorithm \ref{alg:TranMin} the average received power of each E-R always remains constant and matches the desired power level of the E-Rs. However, as we increase the number of antennas in the transmitter, the transmitted power is more focused on the receivers, and therefore, as can be seen in Fig. \ref{fig:MinTranTwoER}, the required power at the transmitter is decreased.

\begin{figure}
\centering
\includegraphics[height = 0.6\linewidth]{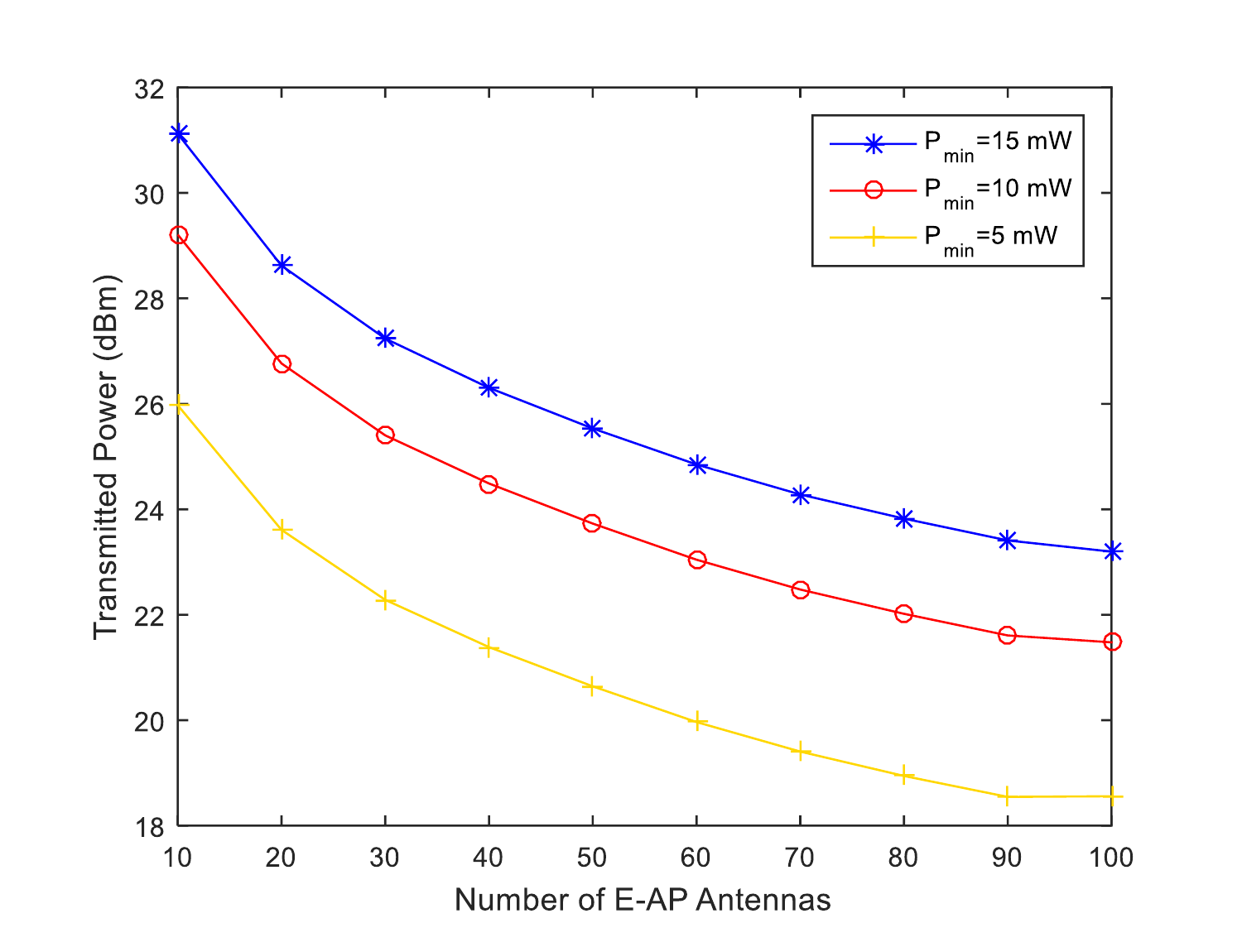}
\caption{Transmitted power of the E-AP versus the number of the E-AP antennas}
\label{fig:MinTranTwoER}
\end{figure}

We also compare the performance of the optimal solution to the MDPP solution. Fig. \ref{fig:MinTranOneER} shows the required transmit power versus the number of transmit antennas when the target received power of each E-R is 15 mW. As can be seen from this figure, the proposed near-optimal algorithm, i.e., Algorithm \ref{alg:TranMin}, performs very well, and its performance is within 5 percent of the optimal solution. While, unlike the optimal solution, Algorithm \ref{alg:TranMin} does not require the statistical information of the distribution of the CSI.

\begin{figure}
\centering
\includegraphics[height = 0.6\linewidth]{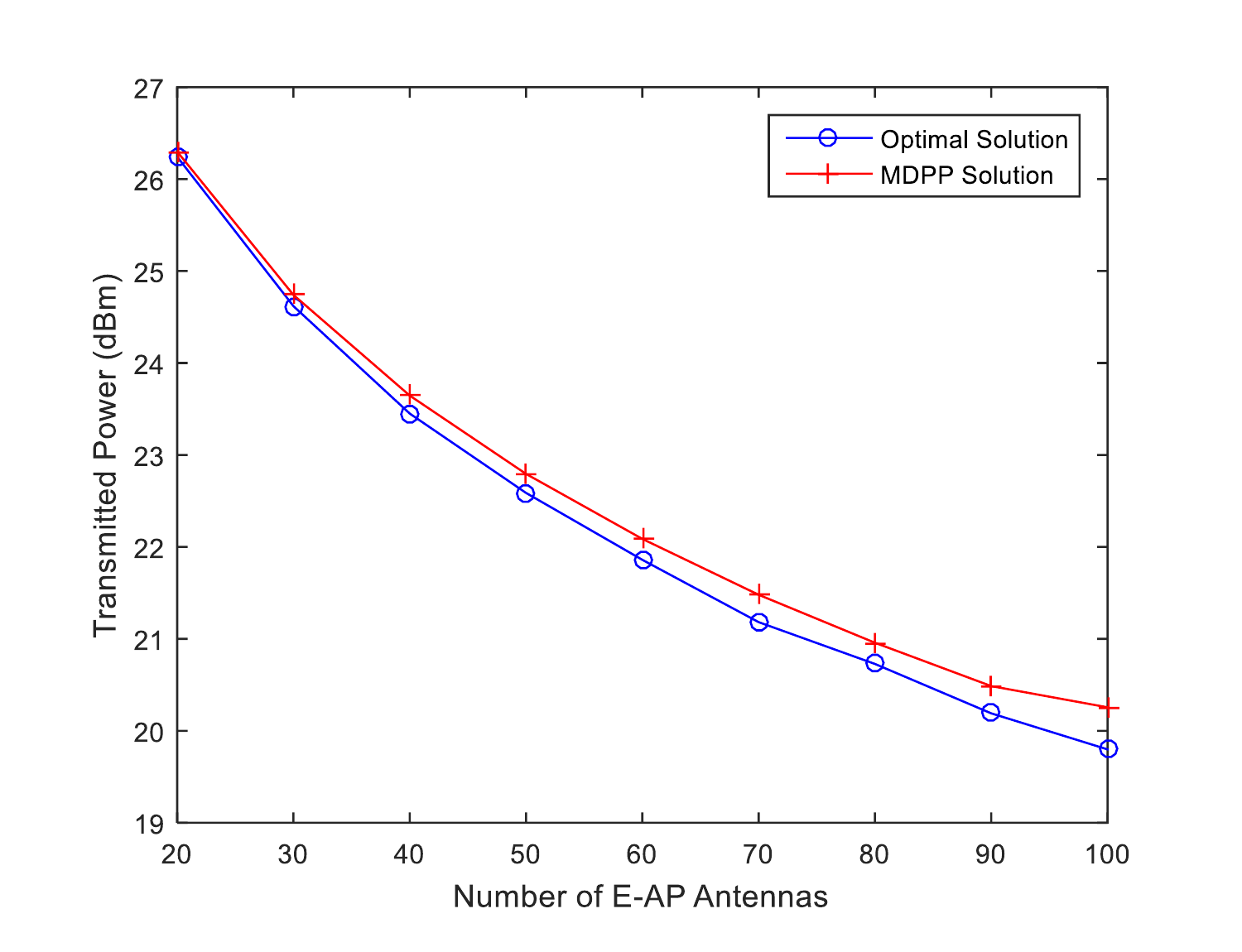}
\caption{Comparison of the near-optimal Algorithm \ref{alg:TranMin} and the optimal solution when we have one E-R}
\label{fig:MinTranOneER}
\end{figure}
\subsection{Power-Limited Case}
As mentioned earlier, in the power-limited case, the E-AP is connected to a stable source of energy and we aim at maximizing the received power of the E-Rs.
Fig. \ref{fig:recPNoMin} shows the average received power at each
E-R versus the number of the E-AP's antennas, and compare the near optimal Algorithm \ref{alg:RecvMax} and the optimal algorithm. Limitations for the E-AP are considered to be 5 and 10 W, respectively. First note that as expected, without considering any fairness models, almost all the transmitted power is delivered to the first E-R, which is closer to the E-AP. Moreover, it can also be seen from this figure that the performance of Algorithm \ref{alg:RecvMax} is still very close to that of the optimal solution.

\begin{figure}
\centering
\includegraphics[height = 0.5\linewidth]{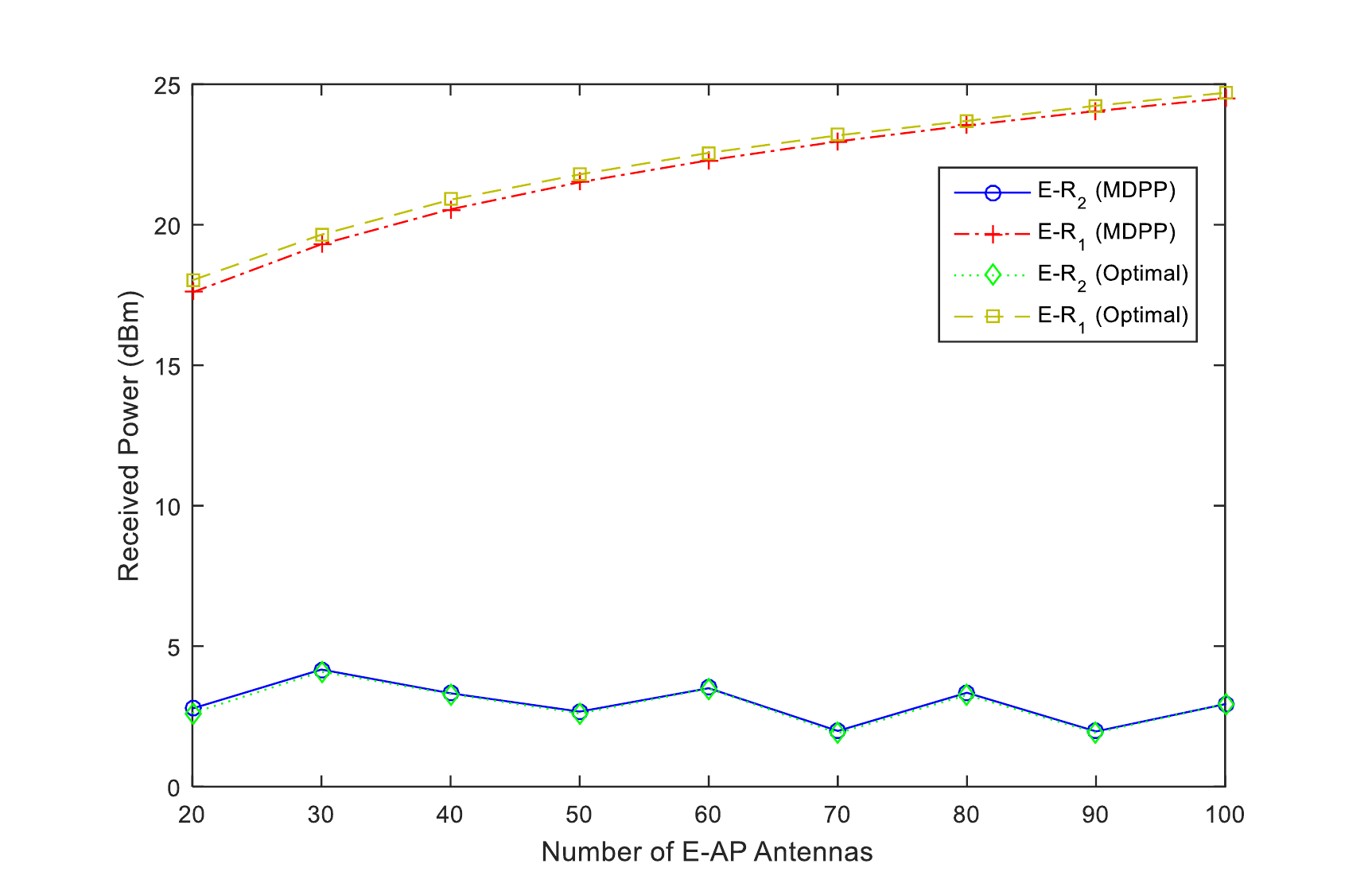}
\caption{The average received power of each
E-R versus the number of the E-AP's antennas under the near-optimal Algorithm \ref{alg:RecvMax} (MDPP) and the optimal solution.}
\label{fig:recPNoMin}
\end{figure}

The performance of the MMF and QPF policies, which consider fairness among the E-Rs, is shown in Fig. \ref{fig:FairUsers}. The MMF policy tries to transfer an equal amount of power to the receivers, while the QPF  policy makes a trade-off between fairness among the E-Rs and their total received energy.
 The maximum and average transmit power of the E-AP are considered to be 10 and 5 W, respectively. Fig. \ref{fig:FairUsers}-(a) shows the average received power of each E-R versus the distance ratio of the E-Rs (denoted by $d_r$), which is defined as $d_r \triangleq \frac{d_f}{d_c}$, where $d_f$ and $d_c$ are the distances of the E-AP to the farther E-R and the closer E-AP, respectively.
To increase $d_r$, move the farther E-R upward away from the E-AP. First note that when $d_r$ equals one, then, the  E-Rs receive the same amount of power, as expected. Unlike the previous policies without fairness that devote almost all of its power to the closer E-R even when $d_r$ is a little more than one, the  MMF policy devotes an equal amount of power to both E-Rs irrespective of the $d_r$ value. However, such approach may lead to a drastic degradation in the total performance when the distances of the E-Rs to the E-AP are too different. In contrast, the QPF policy decreases  the amount of power of the farther E-R smoothly as a function of $d_r$. Hence, the QPF policy leads to a smooth increasing of difference in the amount of received power of the E-Rs when $d_r$ increases. On the other hand, the E-AP provides the required power level of the farther E-R if $d_r$ is much larger than one.


\begin{figure}
\centering
\includegraphics[height = 0.4\linewidth]{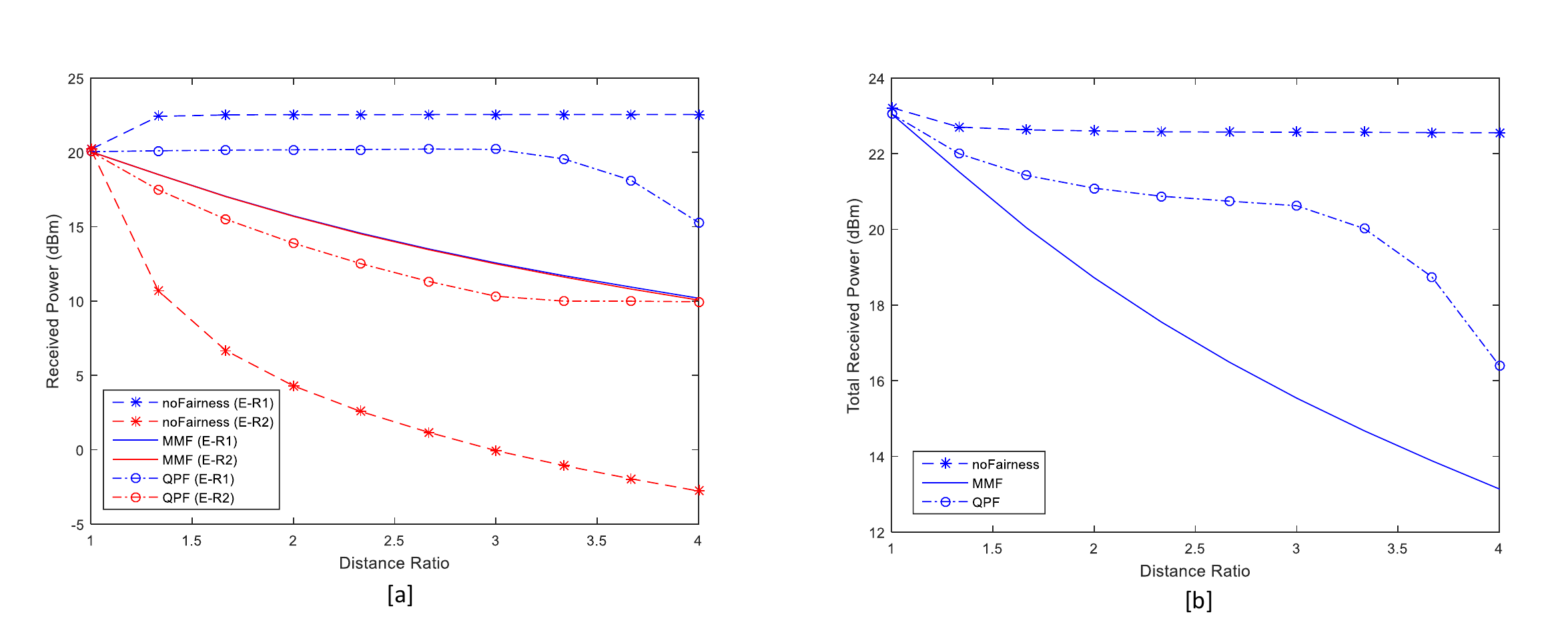}
\caption{(a) The average received power of E-Rs in case of no fairness, MMF and QPF   \ \ \ \ \ \   \ \        (b) Total average received power of E-Rs in case of no fairness, MMF and QPF versus the distance ratio }
\label{fig:FairUsers}
\end{figure}

Fig. \ref{fig:FairUsers}(b) shows the total average received power ($P_{R,T}$) of the E-Rs, and compares the policies with and without fairness. It can be seen from this figure that when considering fairness (either by MMF or QPF schemes), the value of $P_{R,T}$ reduces with the increase in the distance ratio. More specifically, under the MMF policy, which is a strict fair policy, the total received power is minimized. Moreover,  $P_{R,T}$ in both fairness models is a monotonically decreasing function of $d_r$.
 Finally, QPF policy makes a balance between the MMF policy and the policy without fairness.

\section{Conclusion}
In this paper, we proposed optimal and near-optimal policies for wireless power transfer by an E-AP to multiple E-Rs, where the E-AP is battery-operated (energy-limited case) or has limitations on its average transmit power (power-limited case). First, we considered a limited energy E-AP that only transfers the minimum required energy of the E-Rs. We formulated the problem for the case of one E-AP and one E-R and derived the optimal solution. We also proposed a solution based on MDPP using Lyapunov optimization for the cases when there are many E-Rs, and demonstrated the near-optimal performance of the proposed solution. Numerical results showed that the proposed near-optimal solution is within 5 percent of the optimal solution. We also considered the case when we have an E-AP with limited average transmit power capability, and proposed energy transfer policies under different scenarios. First, we proposed a policy to optimize the total received power of all E-Rs. We then demonstrated that in some cases, this policy may lead to unfair distribution of power between the E-Rs. To resolve this issue, we proposed two policies, MMF and QPF, which provide fairness in energy distribution between the E-Rs. The MMF policy tries to transfer an equal amount of power to receivers at the expense of decreasing the total performance. The QPF policy makes a trade-off between the total received power and fairness among E-Rs. It smoothly decreases the received power of farther E-Rs as their distance from E-AP increase. In addition, it provides the minimum required power for farther E-Rs. Numerical results showed that in case of two E-Rs, the MMF policy transfers equal amount of power to E-Rs. The QPF policy decreases the received power of the farther E-R with a slow slope as it is moved away from the E-AP, while always guaranteeing its minimum power level requirement.

\appendices
\section{Proof of Theorem \ref{theo:MinTran}}
\label{app:MinTran}
It is shown in \cite{Neely2010} that the optimal solution of the problem can be achieved using a policy in which $\mathbf{x}$ is only a function of $\mathbf{H}$ at the current timeslot. In this case, as $\mathbf{H}$ is i.i.d. in successive timeslots, $\mathbf{x}$ is also i.i.d. and we can simplify the formulas of transmitted and received power  and rewrite the problem as follows:
\begin{mini!}|l|
 {\{\mathbf{x}(\mathbf{H})\}}{\mathbb{E}[Tr(\mathbf{xx^\ast})]}{\label{P1A}}{}
\addConstraint{ \mathbb{E}[Tr(\mathbf{W_i}\mathbf{x}\mathbf{x}^\ast)]\ge P^{recv}_i, \forall i =  1,2,\ldots,K} \label{equ:P1ExpConstApp}
\addConstraint{Tr(\mathbf{x}\mathbf{x}^\ast)\le P_{peak}, \forall \mathbf{x}\in \mathcal{X}},
 \end{mini!}
where $\mathcal{X}$ denotes the set of possible transmission vector of transmitter. 

The optimal solution satisfies the constraint of equation \eqref{equ:P1ExpConstApp} as follows:
\begin{align}
\mathbb{E}[Tr(\mathbf{W}_1 \mathbf{x}\mathbf{x}^\ast)] &= \mathbb{E}[\lambda_{max}P_x]=P(\lambda_{max}>\lambda_{Th})P_{peak}\mathbb{E}[\lambda_{max}|\lambda_{max}>\lambda_{Th}] \nonumber \\
&=P_{peak}\int_{\lambda_{Th}}^\infty \alpha f_{\lambda_{max}}(\alpha)d\alpha=P_1^{recv},\ P_x=||\mathbf{x}||^2.
\end{align}
We want to show that if a policy yields $\mathbb{E}[P_z]<\mathbb{E}[P_x]=P_{peak}P(\lambda_{max}>\lambda_{Th})$, then the constraint of equation \eqref{equ:P1ExpConst} is not satisfied. Hence, the proposed solution yields the minimum of the average received energy.

We show that if $\mathbf{z}(\mathbf{H})$ yields a lower average power than the optimal solution, then the constraint of equation \eqref{equ:P1ExpConst} violates:
\begin{align*}
\mathbb{E}&[\lambda_{max}P_z]-P_1^{recv}=\mathbb{E}[\lambda_{max}P_z]-\mathbb{E}[\lambda_{max}P_x]=P(\lambda_{max}\ge\lambda_{Th})\mathbb{E}[\lambda_{max}P_z|\lambda_{max}\ge\lambda_{Th}] \\&+P(\lambda_{max}<\lambda_{Th})\mathbb{E}[\lambda_{max}P_z|\lambda_{max}<\lambda_{Th}]-P(\lambda_{max}\ge\lambda_{Th})\mathbb{E}[\lambda_{max}P_x|\lambda_{max}\ge\lambda_{Th}]\\
& = P(\lambda_{max}\ge\lambda_{Th})\mathbb{E}[\lambda_{max}(P_z-P_{peak})|\lambda_{max}\ge\lambda_{Th}]+P(\lambda_{max}<\lambda_{Th})\mathbb{E}[\lambda_{max}P_z|\lambda_{max}<\lambda_{Th}] \\
&\stackrel{a}{\le} \lambda_{Th}\bigg(P(\lambda_{max}\ge\lambda_{Th})(\mathbb{E}[(P_z-P_{peak})|\lambda_{max}\ge\lambda_{Th}])+P(\lambda_{max}<\lambda_{Th})\mathbb{E}[P_z|\lambda_{max}<\lambda_{Th}]\bigg) \\
 &= \lambda_{Th}(\mathbb{E}[P_z]-\mathbb{E}[P_x]) < 0. 
\end{align*}
{$(a)$ We have $P_z-P_{peak}\le 0$ and $\lambda_{max}\ge\lambda_{Th}$ in the first term, and $P_z\ge 0$ and $\lambda_{max}\le\lambda_{Th}$ in the second term. This completes the proof.}

\section{Proof of Lemma \ref{Lem:DriftMinTran}}
\label{app:Lemma2}
Considering the virtual queues update equation we can write,
\begin{align}
Z_i^2[l+1]&\le Z_i^2[l]+ Q^{d}_i[l] ^2+2Z_i[l] Q^{d}_i[l] \nonumber  \\ \Rightarrow  \Delta(\mathbf{Z}[l]) &\le \sum_{i=1}^K Z_i[l]\mathbb{E}[Q^{d}_i[l] |\mathbf{Z}[l]] + \frac{1}{2}\sum_{i=1}^K \mathbb{E}[Q^{d}_i[l]^2|\mathbf{Z}[l]].
\end{align}
The received signal at receiver $i$ has a maximum power of $P_{peak}$. Hence $\frac{1}{2}\sum_{i=1}^K \mathbb{E}[Q^{d}_i[l] ^2|\mathbf{Z}[l]] < \frac{K}{2}P_{peak}^2 = B$.

\section{Proof of Theorem \ref{theo:MaxRecv}}
\label{app:OptRecvMax}
By using similar arguments as in Appendix \ref{app:MinTran}, the problem can be reformulated as follows:
\begin{maxi!}|l|
 {\{\mathbf{x}(\mathbf{H})\}}{\mathbb{E}[P_xTr(\mathbf{W'}\mathbf{\tilde{x}}\mathbf{\tilde{x}}^\ast)]}{\label{P2A}}{}\label{equ:P2ObjFunc}
\addConstraint{ \mathbb{E}[P_x]\le P_{avg}} \label{equ:P2ConstrPavg}
\addConstraint{P_x\le P_{peak}}
\addConstraint{||\mathbf{\tilde{x}}||_2 = 1}, \label{equ:Norm2Eq}
 \end{maxi!}
where $\mathbf{x}=P_x \mathbf{\tilde{x}}$. As mentioned in Lemma \ref{Lemma:baseLemma}, the maximum of equation \eqref{equ:P2ObjFunc} is $\mathbb{E}[P_x \lambda_{max}^{\mathbf{w'}}]$ and this maximum is acheived when $\mathbf{\tilde{x}}=\mathbf{u}_{max}^{\mathbf{w'}}$. The optimal solution satisfies the constraint of equation \eqref{equ:P2ConstrPavg} with the equality as follows:
\begin{align*}
\mathbb{E}[P_x]=P_{peak}P(\lambda_{max}^{w'}\ge \lambda_{Th}^{w'})=P_{peak}(1-F_{\lambda_{max}^{w'}}(\lambda_{Th}^{w'}))=P_{avg}.
\end{align*}

We show that no policy yields a higher average received power than the optimal solution. We show that if $\mathbf{z}(\mathbf{H})$ yields a higher average received energy, then the constraint of equation \eqref{equ:P2ConstrPavg} violates.
\begin{align*}
&\mathbb{E}[\lambda_{max}P_z]-\mathbb{E}[\lambda_{max}P_x]=P(\lambda_{max}\ge\lambda_{Th})\mathbb{E}[\lambda_{max}P_z|\lambda_{max}\ge\lambda_{Th}] \\&+P(\lambda_{max}<\lambda_{Th})\mathbb{E}[\lambda_{max}P_z|\lambda_{max}<\lambda_{Th}]-P(\lambda_{max}\ge\lambda_{Th})\mathbb{E}[\lambda_{max}P_x|\lambda_{max}\ge\lambda_{Th}]\\
& = P(\lambda_{max}\ge\lambda_{Th})\mathbb{E}[\lambda_{max}(P_z-P_{peak})|\lambda_{max}\ge\lambda_{Th}]+P(\lambda_{max}<\lambda_{Th})\mathbb{E}[\lambda_{max}P_z|\lambda_{max}<\lambda_{Th}] \\
&\le \lambda_{Th}(P(\lambda_{max}\ge\lambda_{Th})(\mathbb{E}[(P_z-P_{peak})|\lambda_{max}\ge\lambda_{Th}])+P(\lambda_{max}<\lambda_{Th})\mathbb{E}[P_z|\lambda_{max}<\lambda_{Th}]) \\
 &= \lambda_{Th}(\mathbb{E}[P_z]-\mathbb{E}[P_x]) = \lambda_{Th}(\mathbb{E}[P_z]-P_{avg}) \le 0.
\end{align*}

This completes the proof.

\end{document}